\documentclass{article}
\usepackage{amsmath,amssymb,subfigure,url,amsthm,mathpazo}
\usepackage[dvipdfmx]{graphicx}
\allowdisplaybreaks[1]
\newtheorem{thm}{Theorem}[section]

\newtheorem{lem}[thm]{Lemma}
\newtheorem{prop}[thm]{Proposition}

\newtheorem{rem}[thm]{Remark}
\newtheorem{ass}[thm]{Assumption}

\def\fin   {\hfill{$\Box$}\vspace{5mm}}
\def\l     {\left}
\def\r     {\right}
\def\<     {\langle}
\def\>     {\rangle}

\def\calB  {{\cal B}}

\def\calF  {{\cal F}}

\def\bbC   {{\mathbb C}}

\def\bbE   {{\mathbb E}}
\def\bbF   {{\mathbb F}}

\def\bbP   {{\mathbb P}}

\def\bbR   {{\mathbb R}}
\def\ve    {\varepsilon}

\def\tP    {\bbP^\ast}

\def\tN    {\widetilde{N}}

\begin{document}

\title{Numerical analysis on local risk-minimization
       for exponential L\'{e}vy models}
\author{Takuji Arai\footnote{
        Department of Economics, Keio University, 2-15-45 Mita, Minato-ku,
        Tokyo, 108-8345, Japan, 
        Email: arai@econ.keio.ac.jp}, 
        Yuto Imai\footnote{Department of Mathematics, Waseda University, 3-4-1 Okubo,          Shinjyuku-ku, Tokyo, 169-8555, Japan, E-mail: y.imai@aoni.waseda.jp} and 
        Ryoichi Suzuki\footnote{
        Department of Mathematics, Keio University, 3-14-1 Hiyoshi Kohoku-ku,
        Yokohama, 223-8522, Japan, E-mail:reicesium@gmail.com}
}
\maketitle

%%%%%%%%%%%%%%%%%%%%%%%%%%%%%%%%%%%%%%%%%%%%%%%%%%%%%%%%%%%%%%%%%%%%%%%%%%%%%%%
\begin{abstract}
We illustrate how to compute local risk minimization (LRM)
of call options for exponential L\'evy models.
We have previously obtained a representation of LRM for call options;
here we transform it into a form that allows use of
the fast Fourier transform method suggested by Carr \& Madan.
In particular, we consider Merton jump-diffusion models and
variance gamma models as concrete applications.
\end{abstract}

%%%%%%%%%%%%%%%%%%%%%%%%%%%%%%%%%%%%%%%%%%%%%%%%%%%%%%%%%%%%%%%%%%%%%%%%%%%%%%%
%                                                                             %
% section 1                                                                   %
%                                                                             %
%%%%%%%%%%%%%%%%%%%%%%%%%%%%%%%%%%%%%%%%%%%%%%%%%%%%%%%%%%%%%%%%%%%%%%%%%%%%%%%
\setcounter{equation}{0}
\section{Introduction}
Local risk minimization (LRM), which has more than 20 years' history,
is a well-known hedging method for contingent claims in incomplete markets.
Although its theoretical aspects have been very well studied,
corresponding computational methods have yet to be thoroughly developed.
This paper aims to illustrate how to numerically calculate
LRM for call options in exponential L\'evy models.
To our knowledge, this contribution is the first to address this subject.
In Arai \& Suzuki \cite{AS}, we obtained a representation of LRM for call options
by using Malliavin calculus for L\'evy processes based
on the canonical L\'evy space.
Here we transform that result into a form that allows 
the fast Fourier transform method suggested by Carr \& Madan \cite{CM} to be applied.
In particular, Merton jump-diffusion and variance gamma models,
being common classes of exponential L\'evy models, are discussed as concrete
applications of our approach.

Consider a financial market composed of one risk-free asset and
one risky asset with finite time horizon $T>0$.
For simplicity, we assume that the market's interest rate is zero,
that is, the price of the risk-free asset is 1 at all times.
The fluctuation of the risky asset is assumed to be described by
an exponential L\'evy process $S$ on a complete probability space
$(\Omega, \calF, \bbP)$,\footnote{
$(\Omega, \calF, \bbP)$ is taken as the product of a one-dimensional Wiener
space and the canonical L\'evy space for $N$.
Moreover, we take $\bbF=\{\calF_t\}_{t\in[0,T]}$ as
the completed canonical filtration for $\bbP$.
For more details on the canonical L\'evy space, see \cite{S07} and \cite{AS}.}
described by
\[
S_t:=S_0\exp\l\{\mu t+\sigma W_t+\int_{\bbR_0}x\tN([0,t],dx)\r\} \quad
\mbox{ for }t\in[0,T]\,,
\]
where $S_0>0$, $\mu\in\bbR$, $\sigma>0$, and $\bbR_0:=\bbR\setminus\{0\}$.
Here $W$ is a one-dimensional Brownian motion
and $\tN$ is the compensated version of a Poisson random measure $N$.
Denoting the L\'evy measure of $N$ by $\nu$,
we have $\tN([0,t],A)=N([0,t],A)-t\nu(A)$ for any $t\in[0,T]$ and
$A\in\calB(\bbR_0)$.
Moreover, $S$ is also a solution to the stochastic differential
equation
\[
dS_t=S_{t-}\bigg[\mu^S\,dt+\sigma \,dW_t+\int_{\bbR_0}(e^x-1)\tN(dt,dx)\bigg],
\]
where $\mu^S:=\mu+\frac{1}{2}\sigma^2 +\int_{\bbR_0}(e^x-1-x)\nu(dx)$.
Without loss of generality, we may assume that $S_0=1$ for simplicity.
Now, defining $L_t:=\log S_t$ for all $t\in[0,T]$,
we obtain a L\'evy process $L$.
Moreover, 
$dM_t:=S_{t-}\big[\sigma \,dW_t+\int_{\bbR_0}(e^x-1)\tN(dt,dx)\big]$
is the martingale part of $S$.

Our focus is the development of a computational method
for LRM with respect to a call option $(S_T-K)^+$ with strike price $K>0$.
We do not review the definition of LRM in this paper; for details, see
Schweizer (\cite{Sch}, \cite{Sch3}).
We first briefly introduce the explicit LRM representation of such options
in exponential L\'evy models given in \cite{AS}.

Define the minimal martingale measure $\tP$ as
an equivalent martingale measure under which any square-integrable
$\bbP$-martingale orthogonal to $M$ remains a martingale.
Its density is then given by
\[
\frac{d\tP}{d\bbP}
=\exp\l\{-\xi W_T-\frac{\xi^2}{2}T+\int_{\bbR_0}\log(1-\theta_x)N([0,T],dx)
 +T\int_{\bbR_0}\theta_x\nu(dx)\r\},
\]
where
\[
\xi:=\frac{\mu^S\sigma}{\sigma^2+\int_{\bbR_0}(e^y-1)^2\nu(dy)}
\quad\mbox{ and }\quad
\theta_x:=\frac{\mu^S(e^x-1)}{\sigma^2+\int_{\bbR_0}(e^y-1)^2\nu(dy)}
\]
for $x\in\bbR_0$.
In the development of our approach, we rely on the following:

%%%%%%%%%%%%%%%%%%%%%%%%%%%%%%%%%%%%%%%%%%%%%%%%%%%%%%%%%%%%%%%%%%%%%%%%%%%%%%%
\begin{ass}
\label{ass-1}
\begin{enumerate}
\item $\int_{\bbR_0}(|x|\vee x^2)\nu(dx)<\infty$, and
      $\int_{\bbR_0}(e^x-1)^n\nu(dx)<\infty$ for $n=2,4$.
\item $0\geq\mu^S>-\sigma^2-\int_{\bbR_0}(e^x-1)^2\nu(dx)$.
\end{enumerate}
\end{ass}

\noindent
The first condition ensures that $\mu^S$, $\xi$, and $\theta_x$
are well defined, the square integrability of $L$, and
the finiteness of $\int_{\bbR_0}(e^x-1)^n\nu(dx)$ for $n=1,3$.
The second guarantees that $\theta_x<1$ for any $x\in\bbR_0$.
Moreover, by the Girsanov theorem, $W_t^{\tP}:=W_t+\xi t$ and
$\tN^{\tP}([0,t],dx):=\theta_x\nu(dx)t+\tN([0,t],dx)$
are a $\tP$-Brownian motion and the compensated Poisson random measure of $N$
under $\tP$, respectively.
We can then rewrite $L_t$ as
$L_t=\mu^* t+\sigma W_t^{\tP}+\int_{\bbR_0}x\tN^{\tP}([0,t],dx)$,
where $\mu^*:=-\frac{1}{2}\sigma^2 +\int_{\bbR_0}(x-e^x+1)(1-\theta_x)\nu(dx)$.
Note that $L$ is a L\'evy process even under $\tP$,
with L\'evy measure given by $\nu^{\tP}(dx):=(1-\theta_x)\nu(dx)$.
The LRM will be given as a predictable process $LRM_t$,
which represents the number of units of the risky asset the investor holds
at time $t$.
First, we define
\begin{align}
\label{eq-def-I1}
I_1&:= \bbE_{\tP}[{\bf 1}_{\{S_T>K\}}S_T\mid\calF_{t-}]\,, \\
\label{eq-def-I2}
I_2&:= \int_{\bbR_0}\bbE_{\tP}[(S_Te^x-K)^+-(S_T-K)^+\mid\calF_{t-}]
       (e^x-1)\nu(dx)\,.
\end{align}
Our explicit representation of LRM for call option $(S_T-K)^+$
is then as follows:

%%%%%%%%%%%%%%%%%%%%%%%%%%%%%%%%%%%%%%%%%%%%%%%%%%%%%%%%%%%%%%%%%%%%%%%%%%%%%%%
\begin{prop}[Proposition 4.6 of \cite{AS}]
\label{prop-AS}
For any $K>0$ and $t\in[0,T]$, 
\begin{equation}
\label{eq-prop-AS}
LRM_t=\frac{\sigma^2I_1+I_2}{S_{t-}\big(\sigma^2
      +\int_{\bbR_0}(e^x-1)^2\nu(dx)\big)}.
\end{equation}
\end{prop}

\begin{rem}
\begin{enumerate}
\item
The assumption $\int_{\bbR_0}(e^x-1)^4\nu(dx)<\infty$ is imposed
in Proposition~4.6 of \cite{AS}.
\item
If the interest rate of our market is instead $r>0$, then
(\ref{eq-prop-AS}) becomes
\[
LRM_t=e^{-r(T-t)}\frac{\sigma^2I_1+I_2}{S_{t-}\big(\sigma^2
      +\int_{\bbR_0}(e^x-1)^2\nu(dx)\big)},
\]
and $\tP$ is rewritten with $\xi$ and $\theta_x$ becoming
$\frac{(\mu^S-r)\sigma}{\sigma^2+\int_{\bbR_0}(e^y-1)^2\nu(dy)}$ \\
and $\frac{(\mu^S-r)(e^x-1)}{\sigma^2+\int_{\bbR_0}(e^y-1)^2\nu(dy)}$,
respectively.
Moreover, the second condition in Assumption~\ref{ass-1} would be revised to
$0\geq\mu^S-r>-\sigma^2-\int_{\bbR_0}(e^x-1)^2\nu(dx)$.
That is, a nonzero $r$ requires only that we replace $\mu$ with $\mu-r$
and multiply the the expression for $LRM_t$ by $e^{-r(T-t)}$,
which means that we can easily generalize results for the $r=0$ case to
those for $r>0$.
For simplicity, in this paper we treat only the case $r=0$.
\end{enumerate}
\end{rem}

From the point of view of Proposition~\ref{prop-AS}, we have to calculate
conditional expectations of functionals of $S_T$ under $\tP$
in order to calculate $LRM_t$ numerically.
However, there does not appear to be any straightforward way to specify
the probability density function of $S_T$ (or equivalently $L_T$) under $\tP$.
Instead, since $L$ is a L\'evy process,
it may be comparatively easy to specify its characteristic function
under $\tP$.
Hence, a numerical method based on the Fourier transform is appropriate
for computing LRM.
Moreover, Carr \& Madan \cite{CM} introduced a numerical method for valuing options
based on the fast Fourier transform (FFT).
We take advantage of this to develop a numerical method for LRM.
To this end, we induce integral expressions for $I_1$ and $I_2$
in terms of the characteristic function of $L_{T-t}$ under $\tP$
and recast them into a form that allows
the Carr--Madan approach to be applied.
In particular, $I_2$ will be given as a linear combination of 
Fourier transforms.

In this paper, we consider two concrete exponential L\'evy processes for $L$.
The first is a jump-diffusion process
as introduced by Merton \cite{Mer}.\footnote{
Merton \cite{Mer} also suggested a hedging method for these models,
but this is is different from LRM.
For additional details, see Section~10.1 of \cite{ConT}.}
This consists of a Brownian motion and compound Poisson jumps
with normally distributed jump sizes.
The second is a variance gamma process,
which is a L\'evy process with infinitely many jumps in any finite time
interval and no Brownian component.
This was introduced by \cite{MS} and
can be defined as a time-changed Brownian motion.
Many papers (e.g., \cite{CM}, \cite{MCC}) have studied it in the context of asset prices.
Schoutens \cite{Scho} provides more details on these two L\'evy processes and
more examples of exponential L\'evy models.

There is great deal of literature on numerical experiments related to LRM
(e.g., \cite{BLOS}, \cite{ENS}, \cite{KL}, \cite{LS}, \cite{LVV}, \cite{YEH} ),
but to our knowledge, ours is the first attempt to develop
an FFT-based numerical LRM scheme for exponential L\'evy models.
K\'{e}lani \& Quittard-Pinon \cite{KO} studied an optimal hedging strategy that they
call $\theta$-\textit{hedging}, which is similar to
but different from LRM, for exponential L\'evy models, and adopted
a Fourier transform approach separate from Carr \& Madan \cite{CM}'s method.
As an important difference, they assumed that $S$ is a martingale
under the underlying probability measure.
In contrast, we do not make this assumption.
We therefore need to treat $S$ under $\tP$, that is, calculate conditional
expectations of functionals of $S$ under $\tP$.
However, the structure of $S$ is no longer preserved under a change of measure.
For example, when $L$ is a variance gamma process under $\bbP$,
it is not so under $\tP$.
Thus, our setting is more challenging but also more natural.

The rest of this paper is organized as follows:
An introductory review of the Carr--Madan approach is given in Subsection~2.1,
and the integral representations of $I_1$ and $I_2$ are presented
in Subsection~2.2.
Merton jump-diffusion models are examined in Section~3, which starts
with mathematical preliminaries and proceeds to numerical results.
Section~4 is similarly devoted to variance gamma models.

%%%%%%%%%%%%%%%%%%%%%%%%%%%%%%%%%%%%%%%%%%%%%%%%%%%%%%%%%%%%%%%%%%%%%%%%%%%%%%%
%                                                                             %
% section 2                                                                   %
%                                                                             %
%%%%%%%%%%%%%%%%%%%%%%%%%%%%%%%%%%%%%%%%%%%%%%%%%%%%%%%%%%%%%%%%%%%%%%%%%%%%%%%
\section{Preliminaries}
\subsection{Numerical method}
We briefly review the Carr--Madan approach,
which is an FFT-based numerical approach for option pricing.
The FFT, introduced by \cite{CT},
is a numerical method for computing a discrete Fourier transform given by
\begin{equation}
\label{eq-FFT}
F(l):=\sum^{N-1}_{j=0}e^{-i(2\pi/N)jl}x_j
\end{equation}
for $l=0,\ldots,N-1$,
where $\{x_j\}_{j=0,\ldots,N-1}$ is a sequence on $\bbR$ and where
$N$ is typically a power of $2$.
The FFT requires only $O(N\log_2N)$ arithmetic operations,
as compared with the usual Fourier transform method's
$O(N^2)$.

The aim of the Carr--Madan approach is efficient calculation of
$\bbE[(S_T-K)^+]$ when $S$ is a $\bbP$-martingale.
Recall that we are considering only the case in which the interest rate
is zero.
Denoting $k:=\log K$ and $C(k):=\bbE[(S_T-e^k)^+]$,
we have
\begin{equation}
\label{eq-CM}
C(k)=\frac{1}{\pi}\int_0^\infty e^{-i(v-i\alpha)k}
\frac{\phi(v-i\alpha-i)}{i(v-i\alpha)[i(v-i\alpha)+1]}\,dv
\end{equation}
for $\alpha>0$ with $\bbE[S_T^{\alpha+1}]<\infty$,
where $\phi$ is the characteristic function of $L_T$.
Note that the right-hand side of (\ref{eq-CM}) is independent of
the choice of $\alpha$.
Now, we denote $\psi(z):=\frac{\phi(z-i)}{iz(iz+1)}$ for $z\in\bbC$.
Using the trapezoidal rule, we can therefore approximate $C(k)$ as
\begin{equation}
\label{eq-CM2}
C(k)\approx\frac{1}{\pi}\sum^{N-1}_{j=0}e^{-i(\eta j-i\alpha)k}
           \psi(\eta j-i\alpha)\eta\,,
\end{equation}
where $N$ represents the number of grid points
and $\eta>0$ is the distance between adjacent grid points.
The right-hand side of (\ref{eq-CM2}) corresponds to the integral
in (\ref{eq-CM}) over the interval $[0,N\eta]$,
so we need to specify $N$ and $\eta$ such that
\begin{equation}
\label{eq-CM3}
\l|\frac{1}{\pi}\int_{N\eta}^\infty e^{-i(v-i\alpha)k}\psi(v-i\alpha)dv\r|<\ve
\end{equation}
for a sufficiently small value $\ve>0$,
which represents the allowable error.
By incorporating Simpson's rule weightings,
we may rewrite (\ref{eq-CM2}) as
\[
C(k)\approx \frac{1}{\pi}\sum^{N-1}_{j=0}e^{-i(\eta j-i\alpha)k}
            \psi(\eta j-i\alpha)\frac{\eta}{3}(3+(-1)^{j+1}-\delta_j),
\]
where $\delta_j$ is the Kronecker delta function.
We define
\[
F(l):=\frac{e^{-\alpha k}}{\pi}\sum^{N-1}_{j=0}e^{-i\frac{2\pi}{N}jl}
      e^{i\pi j}\psi(\eta j-i\alpha)\frac{\eta}{3}(3+(-1)^{j+1}-\delta_j)
\]
for $l=0,\ldots,N-1$,
which is a discrete Fourier transform as given in (\ref{eq-FFT}). This yields
\[
C(k)\approx F\l(\!\l(k+\frac{\pi}{\eta}\r)\frac{N\eta}{2\pi}\r).
\]
So long as we take $\eta$ so that $|k|<\pi/\eta$,
we can employ the FFT to compute $C(k)$.

%%%%%%%%%%%%%%%%%%%%%%%%%%%%%%%%%%%%%%%%%%%%%%%%%%%%%%%%%%%%%%%%%%%%%%%%%%%%%%%
\subsection{Integral representations}
%%%%%%%%%%%%%%%%%%%%%%%%%%%%%%%%%%%%%%%%%%%%%%%%%%%%%%%%%%%%%%%%%%%%%%%%%%%%%%%
We next induce integral expressions for $I_1$ and $I_2$,
defined in (\ref{eq-def-I1}) and (\ref{eq-def-I2}),
and evolve them so that the Carr--Madan approach
is available.
Recall that Assumption~\ref{ass-1} applies throughout.
As can be seen from Subsection~2.1, if $I_1$ and $I_2$ 
are represented in the same form as (\ref{eq-CM})
we can compute them by means of the Carr--Madan approach.
Because the conditional expectations appearing in $I_1$ and $I_2$ are
under $\tP$,
the functions corresponding to $\psi$ in (\ref{eq-CM}) should include
the characteristic function of $L_{T-t}$ under $\tP$,
denoted by $\phi_{T-t}(z):=\bbE_{\tP}[e^{izL_{T-t}}]$ for $z\in\bbC$.

First, we induce an integral representation for
$I_1$ (${=}\bbE_{\tP}[{\bf 1}_{\{S_T>K\}}S_T\mid\calF_{t-}]$) with $\phi_{T-t}$
by using Proposition~2 from \cite{Tankov}:

%%%%%%%%%%%%%%%%%%%%%%%%%%%%%%%%%%%%%%%%%%%%%%%%%%%%%%%%%%%%%%%%%%%%%%%%%%%%%%%
\begin{prop}
\label{prop-IR1}
For $K>0$,
\begin{equation}
\label{eq-prop-IR1}
\bbE_{\tP}[{\bf 1}_{\{S_T>K\}}\cdot S_T\mid\calF_{t-}]
= \frac{1}{\pi}\int_0^\infty\frac{K^{-iv-\alpha+1}}{\alpha-1+iv}
  \phi_{T-t}(v-i\alpha)S_{t-}^{\alpha+iv}\,dv
\end{equation}
for all $t\in[0,T]$ and $\alpha\in(1,2]$.
Note that the right-hand side is independent of the choice of $\alpha$.
\end{prop}

\begin{proof}
Define $G(x):={\bf 1}_{\{x>K\}}\cdot x$,
$g(x):=G(e^x)$ for any $x\in\bbR$,
and $\hat{g}(z):=\int_{\bbR}e^{izx}g(x)dx$ for any $z\in\bbC$.
We employ one lemma:

%%%%%%%%%%%%%%%%%%%%%%%%%%%%%%%%%%%%%%%%%%%%%%%%%%%%%%%%%%%%%%%%%%%%%%%%%%%%%%%
\begin{lem}
\label{lem-IR}
Let $L^\prime$ be an independent copy of $L$.
Then, $L^\prime_{T-t}+L_{t-}\overset{\tP\mbox{-}d}{=}L_T$
for all $t\in [0,T]$,
where $A\overset{\tP\mbox{-}d}{=}B$ means that $A=B$ in law for $\tP$.
\end{lem}

\noindent
{\bf Proof of Lemma \ref{lem-IR}.}
Proposition~I.7 of \cite{Ber} implies that $\tP(L_{t-}=L_t)=1$.
Therefore, $L_t\overset{\tP\mbox{-}d}{=}L_{t-}$.
Because L\'evy processes have independent and stationary increments, we have
$L_T=L_T-L_t+L_t\overset{\tP\mbox{-}d}{=}L^\prime_{T-t}+L_t$.
\fin

\noindent
Returning to the proof of Proposition~\ref{prop-IR1},
from Lemma~\ref{lem-IR} we have
\begin{align*}
\bbE_{\tP}[{\bf 1}_{\{S_T>K\}}\cdot S_T\mid\calF_{t-}]
&= \bbE_{\tP}[G(S_T)\mid\calF_{t-}]
   =\bbE_{\tP}[g(L^\prime_{T-t}+L_{t-})\mid\calF_{t-}] \\
&= \int_{\bbR}g(x+L_{t-})p(dx),
\end{align*}
where $p(A):=\tP(L^\prime_{T-t}\in A)$ for any $A\in\calB(\bbR)$.
By (22)--(25) in the proof of Proposition~2 of \cite{Tankov},
if any $\alpha\in(1,2]$ satisfies the conditions that
\begin{description}
 \item[(a)] $g(x)e^{-\alpha x}$ has finite variation on $\bbR$,
 \item[(b)] $g(x)e^{-\alpha x}\in L^1(\bbR)$,
 \item[(c)] $\bbE_{\tP}[e^{\alpha L_{T-t}}]<\infty$, and
 \item[(d)] $\int_{\bbR}\frac{|\phi_{T-t}(v-i\alpha)|}{1+|v|}dv<\infty$,
\end{description}
then
\[
\int_{\bbR}g(x+L_{t-})p(dx)=\frac{1}{2\pi}\int_{\bbR}\hat{g}(v+i\alpha)
                            \phi_{T-t}(-v-i\alpha)S_{t-}^{\alpha-iv}\,dv
\]
for $\alpha\in(1,2]$, which is independent of the choice of $\alpha$.
As a result, under conditions (a)--(d), we have
\begin{align*}
\bbE_{\tP}[{\bf 1}_{\{S_T>K\}}\cdot S_T\mid\calF_{t-}]
&= \frac{1}{2\pi}\int_{\bbR}\hat{g}(v+i\alpha)
   \phi_{T-t}(-v-i\alpha)S_{t-}^{\alpha-iv}\,dv \\
&= \frac{1}{\pi}\int_0^\infty\hat{g}(-v+i\alpha)
   \phi_{T-t}(v-i\alpha)S_{t-}^{\alpha+iv}\,dv \\
&= \frac{1}{\pi}\int_0^\infty\frac{K^{-iv-\alpha+1}}{\alpha-1+iv}
   \phi_{T-t}(v-i\alpha)S_{t-}^{\alpha+iv}\,dv.
\end{align*}

We need only to confirm that conditions (a)--(d) hold.
Conditions (a) and (b) are obvious.
To demonstrate condition (c),
it suffices to show $S_{T-t}\in L^2(\tP)$ for any $t\in[0,T]$.
Note that we have
\begin{eqnarray*}
\lefteqn{\int_{\bbR_0}(e^x-1)^2\nu^{\tP}(dx)} \\
&\qquad =& \int_{\bbR_0}(e^x-1)^2\nu(dx)
    +\frac{|\mu^S|}{\sigma^2+\int_{\bbR_0}(e^x-1)^2\nu(dx)}
    \int_{\bbR_0}(e^x-1)^3\nu(dx)
    < \infty\,.
\end{eqnarray*}
Because $S$ is a solution to
$dS_t=S_{t-}(\sigma \,dW_t^{\tP}+\int_{\bbR}(e^x-1)\tN^{\tP}(dt,dx))$,
Theorem~117 of \cite{Situ} implies that $\sup_{t\in [0,T]}|S_t|\in L^2(\tP)$.

Next, we show condition (d).
Note that{\small
\begin{eqnarray}
\lefteqn{\phi_{T-t}(v-i\alpha)} \nonumber \\
&\qquad =&\bbE_{\tP}\l[\exp\l\{(iv+\alpha)\l[\mu^*(T-t)+\sigma W_{T-t}^{\tP}
    +\int_{\bbR_0}x\tN^{\tP}([0,T-t],dx)\r]\r\}\r]. \nonumber \\
\label{eq-prop-IR1-2}
\end{eqnarray}}
For the right-hand side, we have
\begin{eqnarray}
\label{eq-prop-IR1-3}
\lefteqn{\l|\bbE_{\tP}\l[\exp\l\{(iv+\alpha)
\int_{\bbR_0}x\tN^{\tP}([0,T-t],dx)\r\}\r]\r|} \nonumber \\
&\quad\leq& \bbE_{\tP}\l[\exp\l\{\alpha
       \int_{\bbR_0}x\tN^{\tP}([0,T-t],dx)\r\}\r]
       <\infty\,,
\end{eqnarray}
because
\begin{align*}
\bbE_{\tP}\big[e^{\alpha L_{T-t}}\big]
&= \bbE_{\tP}\l[\exp\l\{\alpha\l[\mu^*(T-t)+\sigma W_{T-t}^{\tP}
    +\int_{\bbR_0}x\tN^{\tP}([0,T-t],dx)\r]\r\}\r] \\
&= e^{\mu^*(T-t)}\bbE_{\tP}\big[e^{\alpha\sigma W_{T-t}^{\tP}}\big]
    \bbE_{\tP}\l[e^{\alpha\int_{\bbR_0}x\tN^{\tP}([0,T-t],dx)}\r],
\end{align*}
$\bbE_{\tP}\big[e^{\alpha\sigma W_{T-t}^{\tP}}\big]
=\exp\l\{\frac{1}{2}\alpha^2\sigma^2(T-t)\r\}$,
and $\bbE_{\tP}\big[e^{\alpha L_{T-t}}\big]<\infty$.
In addition, we obtain
\begin{equation}
\label{eq-prop-IR1-4}
\l|\bbE_{\tP}[\exp\{(iv+\alpha)\sigma W^{\tP}_{T-t}\}]\r|
=\exp\l\{\frac{(\alpha^2-v^2)\sigma^2(T-t)}{2}\r\}.
\end{equation}
As a result, we have from (\ref{eq-prop-IR1-2})--(\ref{eq-prop-IR1-4})
\[
\int_{\bbR}\frac{|\phi_{T-t}(v-i\alpha)|}{1+|v|}dv
<C\int_{\bbR}\frac{1}{1+|v|}\exp\l\{-\frac{\sigma^2(T-t)}{2}v^2\r\}dv<\infty
\]
for some $C>0$.
This completes the proof of Proposition~\ref{prop-IR1}.
\end{proof}

\noindent
We evolve (\ref{eq-prop-IR1}) into the same form as (\ref{eq-CM}) as follows:
\begin{align}
I_1=\bbE_{\tP}[{\bf 1}_{\{S_T>K\}}\cdot S_T\mid\calF_{t-}]
&= \frac{1}{\pi}\int_0^\infty\frac{K^{-iv-\alpha+1}}{\alpha-1+iv}
   \phi_{T-t}(v-i\alpha)S_{t-}^{\alpha+iv}dv \nonumber \\
&= \frac{e^k}{\pi}\int_0^\infty e^{-i(v-i\alpha)k}\psi_1(v-i\alpha)dv\,
\label{eq-I1}
\end{align}
where $k:=\log K$ and
$\psi_1(z):=\frac{\phi_{T-t}(z)S^{iz}_{t-}}{iz-1}$ for $z\in\bbC$.
Thus, we can compute $I_1$ with the FFT based on Subsection~2.1.

We turn next to $I_2$
(${=}\int_{\bbR_0}\bbE_{\tP}[(S_Te^x-K)^+-(S_T-K)^+\mid\calF_{t-}](e^x-1)
\nu(dx)$).
First, we have the following integral representation:

%%%%%%%%%%%%%%%%%%%%%%%%%%%%%%%%%%%%%%%%%%%%%%%%%%%%%%%%%%%%%%%%%%%%%%%%%%%%%%%
\begin{prop}
For any $K>0$,
\begin{equation}
\label{eq-prop-IR2}
\bbE_{\tP}[(S_T-K)^+\mid\calF_{t-}]
=\frac{1}{\pi}\int_0^\infty K^{-iv-\alpha+1}\frac{\phi_{T-t}(v-i\alpha)
 S_{t-}^{\alpha+iv}}{(\alpha-1+iv)(\alpha+iv)}\,dv
\end{equation}
for any $t\in[0,T]$ and any $\alpha\in(1,2]$.
Note that the right-hand side is independent of the choice of $\alpha$.
\end{prop}

\begin{proof}
We can see this in the same manner as Proposition~\ref{prop-IR1}
but with $G(x)=(x-K)^+$.
\end{proof}

\noindent
Note that (\ref{eq-prop-IR2}) coincides with (\ref{eq-CM}),
where $\alpha-1$ in (\ref{eq-prop-IR2}) corresponds to
$\alpha$ in (\ref{eq-CM}).
Denoting $\psi_2(z):=\frac{\phi_{T-t}(z)S_{t-}^{iz}}{(iz-1)iz}$
for $z\in\bbC$ and $\zeta:=v-i\alpha$, we have
\begin{eqnarray}
\bbE_{\tP}[(S_T-K)^+\mid\calF_{t-}]
&=&\frac{1}{\pi}\int_0^\infty K^{-iv-\alpha+1}
   \frac{\phi_{T-t}(v-i\alpha)S_{t-}^{\alpha+iv}}
   {(\alpha-1+iv)(\alpha+iv)}\,dv \nonumber \\
&=&\frac{1}{\pi}\int_0^\infty K^{-i\zeta+1}
   \frac{\phi_{T-t}(\zeta)S_{t-}^{i\zeta}}{(i\zeta-1)i\zeta}\,dv \nonumber \\
&=&\frac{1}{\pi}\int_0^\infty K^{-i\zeta+1}\psi_2(\zeta)dv
   =:f(K).
\label{eq-I2-1}
\end{eqnarray}
Note that $f(K)$ is computed with the FFT.
Moreover, Fubini's theorem implies
\begin{align}
I_2
&= \int_{\bbR_0}\bbE_{\tP}[(S_Te^x-K)^+-(S_T-K)^+\mid\calF_{t-}]
   (e^x-1)\nu(dx) \nonumber\\
&= \int_{\bbR_0}\l\{e^xf(e^{-x}K)-f(K)\r\}(e^x-1)\nu(dx) \nonumber\\
&= \int_{\bbR_0}\l\{\frac{e^x}{\pi}\int_0^\infty(Ke^{-x})^{-i\zeta+1}
   \psi_2(\zeta)dv-\frac{1}{\pi}\int_0^\infty K^{-i\zeta+1}\psi_2(\zeta)dv\r\}
   (e^x-1)\nu(dx) \nonumber \\
&= \int_{\bbR_0}\l\{\frac{1}{\pi}\int_0^\infty(e^{i\zeta x}-1)K^{-i\zeta+1}
   \psi_2(\zeta)dv\r\}(e^x-1)\nu(dx) \nonumber \\
&= \frac{1}{\pi}\int_0^\infty K^{-i\zeta+1}
   \int_{\bbR_0}(e^{i\zeta x}-1)(e^x-1)\nu(dx)\psi_2(\zeta)dv\,,
\label{eq-I2-2}
\end{align}
which is the same form as (\ref{eq-CM}),
because the integrand of (\ref{eq-I2-2}) is a function of $\zeta$.
However, we cannot compute (\ref{eq-I2-2}) numerically as it stands,
because it is not possible to compute
the integral $\int_{\bbR_0}(e^{i\zeta x}-1)(e^x-1)\nu(dx)$ directly.
Thus, we need to make further model-dependent calculations.
In Sections~3 and 4, respectively, we evolve (\ref{eq-I2-2})
into a linear combination of Fourier transforms for Merton jump-diffusion
models and variance gamma models.

%%%%%%%%%%%%%%%%%%%%%%%%%%%%%%%%%%%%%%%%%%%%%%%%%%%%%%%%%%%%%%%%%%%%%%%%%%%%%%%
\begin{rem}
Regarding $LRM_t$, $I_1$, and $I_2$ as functions of $S_{t-}$ and $K$,
we have \\$I_i(S_{t-},K)/S_{t-}=I_i(1,K/S_{t-})$ for $i=1,2$
by (\ref{eq-prop-IR1}) and (\ref{eq-I2-2}), and
\[
LRM_t(S_{t-},K)
=\frac{\sigma^2I_1(S_{t-},K)+I_2(S_{t-},K)}{S_{t-}
  \big(\sigma^2+\int_{\bbR_0}(e^x-1)^2\nu(dx)\big)}
=\frac{\sigma^2I_1(1,K/S_{t-})+I_2(1,K/S_{t-})}{
  \sigma^2+\int_{\bbR_0}(e^x-1)^2\nu(dx)}
\]
by (\ref{eq-prop-AS}).
As a result, $LRM_t$ is given as a function of $K/S_{t-} =:m_{t-}$,
where $m_{t-}$ is called \textit{moneyness}.
Thus, we denote $LRM_t$ by $LRM_t(m_{t-})$.
As a by-product of this, we can analyze jump impacts on LRM.
If the process $L$ has a jump with size $y\in\bbR_0$ at time $t$,
then the moneyness $m_{t-}$ changes into $m_{t-}e^{-y}$ at the moment
when the jump occurs.
Thus, LRM also changes from $LRM_t(m_{t-})$ to $LRM_t(m_{t-}e^{-y})$.
We can regard the difference $LRM_t(m_{t-}e^{-y})-LRM_t(m_{t-})$ as
a jump impact.
In particular, $LRM_t(e^{-y})-LRM_t(1)$ represents a jump impact
when the option is at the money.
\end{rem}

%%%%%%%%%%%%%%%%%%%%%%%%%%%%%%%%%%%%%%%%%%%%%%%%%%%%%%%%%%%%%%%%%%%%%%%%%%%%%%%
\begin{rem}
Hereafter, we fix $\alpha\in(1,2]$ arbitrarily.
Moreover, we denote $\zeta := v-i\alpha$ for $v\in\bbR$, so
we may regard $\zeta$ as a function of $v$.
\end{rem}

%%%%%%%%%%%%%%%%%%%%%%%%%%%%%%%%%%%%%%%%%%%%%%%%%%%%%%%%%%%%%%%%%%%%%%%%%%%%%%%
%                                                                             %
% section 3                                                                   %
%                                                                             %
%%%%%%%%%%%%%%%%%%%%%%%%%%%%%%%%%%%%%%%%%%%%%%%%%%%%%%%%%%%%%%%%%%%%%%%%%%%%%%%
\setcounter{equation}{0}
\section{Merton Jump-Diffusion Models}
%%%%%%%%%%%%%%%%%%%%%%%%%%%%%%%%%%%%%%%%%%%%%%%%%%%%%%%%%%%%%%%%%%%%%%%%%%%%%%%
We consider the case in which $L$ is given
as a Merton jump-diffusion process,
which consists of a diffusion component with volatility $\sigma>0$ and
compound Poisson jumps with three parameters, $m\in\bbR$, $\delta>0$, and
$\gamma>0$.
Note that $\gamma$ represents the jump intensity and that
the sizes of the jumps are distributed normally with mean $m$ and
variance $\delta^2$.
Thus, its L\'evy measure $\nu$ is given by
\[
\nu(dx)
=\frac{\gamma}{\sqrt{2\pi}\delta}\exp\l\{-\frac{(x-m)^2}{2\delta^2}\r\}dx.
\]
When it desirable to emphasize the parameters, we write $\nu$ as
$\nu[\gamma,m,\delta]$.
Note that the first condition of Assumption~\ref{ass-1} is satisfied
for any $m\in\bbR$, $\delta>0$, and $\gamma>0$.
In addition, the second condition is equivalent to
\[
0\geq\mu+\frac{\sigma^2}{2}+\gamma\l\{\exp\l(m+\frac{\delta^2}{2}\r)-1-m\r\}
\]
and
\[
\mu+\frac{3\sigma^2}{2}+\gamma\l\{\exp(2m+2\delta^2)
-\exp\l(m+\frac{\delta^2}{2}\r)-m\r\}>0.
\]
We consider only the case in which the parameters satisfy
Assumption~\ref{ass-1}.

%%%%%%%%%%%%%%%%%%%%%%%%%%%%%%%%%%%%%%%%%%%%%%%%%%%%%%%%%%%%%%%%%%%%%%%%%%%%%%%
\subsection{Mathematical preliminaries}
%%%%%%%%%%%%%%%%%%%%%%%%%%%%%%%%%%%%%%%%%%%%%%%%%%%%%%%%%%%%%%%%%%%%%%%%%%%%%%%
Our aim here is threefold:
(1) to give an analytic form for $\phi_{T-t}(z)$
    (${:=}\bbE_{\tP}[e^{izL_{T-t}}]$);
(2) to evolve (\ref{eq-I2-2}) into a linear combination of
    three Fourier transforms; and
(3) to give sufficient conditions for $N\eta$ under which (\ref{eq-CM3}) holds
    for a given $\ve>0$.

First, we provide an analytic form of $\phi_{T-t}$.
To this end, we begin by calculating $\nu^{\tP}$.
%%%%%%%%%%%%%%%%%%%%%%%%%%%%%%%%%%%%%%%%%%%%%%%%%%%%%%%%%%%%%%%%%%%%%%%%%%%%%%%
\begin{prop}
\label{prop-Mer1}
We have 
\begin{equation}
\label{eq-prop-Mer1}
\nu^{\bbP^*}(dx)
=\nu [(1+h)\gamma,m,\delta^2](dx)
+\nu\l[-h\gamma\exp\l\{\frac{2m+\delta^2}{2}\r\},m+\delta^2,\delta^2\r](dx),
\end{equation}
where $h:=\frac{\mu^S}{\sigma_2+\int_{\bbR_0}(e^x-1)^2\nu(dx)}$.
\end{prop}

\begin{proof}
By Assumption~\ref{ass-1}, $0\geq h>-1$.
Hence,
\[
\nu^{\bbP^*}(dx)
=(1-\theta_x)\nu(dx)=(1-h(e^x-1))\nu(dx)=(1+h)\nu(dx)-he^x\nu(dx).
\]
Moreover,
\begin{align*}
e^x\nu(dx)
&=\frac{\gamma}{\sqrt{2\pi}\delta}\exp\l\{x-\frac{(x-m)^2}{2\delta^2}\r\}dx \\
&=\frac{\gamma}{\sqrt{2\pi}\delta}
  \exp\l\{-\frac{[x-(m+\delta^2)]^2}{2\delta^2}+\frac{2m+\delta^2}{2}\r\}dx \\
&=\nu\l[\gamma\exp\l\{\frac{2m+\delta^2}{2}\r\},m+\delta^2,\delta^2\r](dx),
\end{align*}
from which (\ref{eq-prop-Mer1}) follows.
\end{proof}

\noindent
Next, we calculate $\phi_{T-t}(\zeta)$ for $t\in[0,T]$.

%%%%%%%%%%%%%%%%%%%%%%%%%%%%%%%%%%%%%%%%%%%%%%%%%%%%%%%%%%%%%%%%%%%%%%%%%%%%%%%
\begin{prop}
\label{prop-Mer2}
For any $t\in[0,T]$ and $v\in\bbR$, with $\zeta:=v-i\alpha$,
\begin{eqnarray*}
\phi_{T-t}(\zeta)
&=&\exp\l\{(T-t)\l[i\zeta\mu^*-\frac{\sigma^2\zeta^2}{2}
   +\int_{\bbR_0}(e^{i\zeta x}-1-i\zeta x)\nu^{\bbP^*}(dx)\r]\r\} \\
&=&\exp\Bigg\{(T-t)\Bigg[i\zeta\mu^*-\frac{\sigma^2\zeta^2}{2}
   +(1+h)\gamma(e^{im\zeta-\frac{\zeta^2\delta^2}{2}}-1-im\zeta) \\
&& \qquad\qquad\qquad\; -h\gamma e^{\frac{2m+\delta^2}{2}}[e^{i(m+\delta^2)
   \zeta
   -\frac{\zeta^2\delta^2}{2}}-1-i\zeta(m+\delta^2)]\Bigg]\Bigg\}.
\end{eqnarray*}
\end{prop}

\begin{proof}
We only have to show the first equality:
{\small
\begin{eqnarray*}
\phi_{T-t}(\zeta)
&=&\bbE_{\tP}\l[\exp\l\{i\zeta\l[\mu^*(T-t)+\sigma W_{T-t}^{\tP}
    +\int_{\bbR_0}x\tN^{\tP}([0,T-t],dx)\r]\r\}\r] \\
&=&\exp\l\{(T-t)i\zeta\mu^*\r\}\bbE_{\tP}[e^{i\zeta\sigma W_{T-t}^{\tP}}]
   \bbE_{\tP}\l[\exp\l\{i\zeta\int_{\bbR_0}x\tN^{\tP}([0,T-t],dx)\r\}\r] \\
&=&\exp\l\{(T-t)\l[i\zeta\mu^*-\frac{\sigma^2\zeta^2}{2}
   +\int_{\bbR_0}(e^{i\zeta x}-1-i\zeta x)\nu^{\bbP^*}(dx)\r]\r\}.
\end{eqnarray*}}
\end{proof}

Second, we evolve (\ref{eq-I2-2}).
We define 
$\tilde{\psi}(z):=\psi_2(z)\exp\l\{-\frac{1}{2}\delta^2 z^2\r\}$
for $z\in\bbC$ and
$\tilde{f}(K):=\frac{1}{\pi}\int_0^\infty K^{-i\zeta+1}
               \tilde{\psi}(\zeta)dv$.
Remark that $\tilde{f}$ is computed with the FFT
as well as $f$ defined in (\ref{eq-I2-1}).
The following proposition demonstrates (\ref{eq-I2-2}), namely,
$I_2$ is given by a linear combination of three Fourier transforms.

%%%%%%%%%%%%%%%%%%%%%%%%%%%%%%%%%%%%%%%%%%%%%%%%%%%%%%%%%%%%%%%%%%%%%%%%%%%%%%%
\begin{prop}
\label{prop-Mer3}
We have
\begin{eqnarray}
\lefteqn{\int_{\bbR_0}\bbE_{\tP}[(S_Te^x-K)^+-(S_T-K)^+\mid\calF_{t-}]
         (e^x-1)\nu(dx)} \nonumber \\
&\qquad =&\gamma e^{2m+\frac{3}{2}\delta^2}\tilde{f}(Ke^{-m-\delta^2})
   -\gamma e^m\tilde{f}(Ke^{-m})+\gamma(1-e^{m+\frac{\delta^2}{2}})f(K)
\label{eq-prop-Mer3}
\end{eqnarray}
for any $t\in[0,T]$.
\end{prop}

\begin{proof}
We calculate
{\small
\begin{eqnarray*}
\lefteqn{\int_{\bbR_0}(e^{i\zeta x}-1)(e^x-1)\nu(dx)} \\
&\qquad=&\int_{\bbR_0}(e^{(i\zeta+1)x}-e^{i\zeta x}+1-e^x)\nu(dx) \\
&\qquad=&\gamma\exp\l\{(i\zeta+1)m+\frac{\delta^2}{2}(i\zeta+1)^2\r\}
   -\gamma\exp\l\{i\zeta m-\frac{\delta^2}{2}\zeta^2\r\}
   +\gamma(1-e^{m+\frac{\delta^2}{2}}).
\end{eqnarray*}}
Hence, we obtain
\begin{eqnarray*}
(\ref{eq-I2-2})
&=&\frac{\gamma}{\pi}e^{m+\frac{\delta^2}{2}}
   \int_0^\infty e^{i\zeta(m+\delta^2)}K^{-i\zeta+1}e^{-\frac{\delta^2}{2}
   \zeta^2}\psi_2(\zeta)dv \\
&& -\frac{\gamma}{\pi}\int_0^\infty(Ke^{-m})^{-i\zeta+1}e^m
   e^{-\frac{\delta^2}{2}\zeta^2}\psi_2(\zeta)dv
   +\gamma(1-e^{m+\frac{\delta^2}{2}})f(K) \\
&=&\gamma e^{2m+\frac{3}{2}\delta^2}\tilde{f}(Ke^{-m-\delta^2})
   -\gamma e^m\tilde{f}(Ke^{-m})+\gamma(1-e^{m+\frac{\delta^2}{2}})f(K).
\end{eqnarray*}
\end{proof}

Third, we provide sufficient conditions for the product $N\eta$
under which (\ref{eq-CM3}) holds for a given allowable error $\ve>0$.
First of all, we determine an upper estimate for $\phi_{T-t}$.

%%%%%%%%%%%%%%%%%%%%%%%%%%%%%%%%%%%%%%%%%%%%%%%%%%%%%%%%%%%%%%%%%%%%%%%%%%%%%%%
\begin{prop}
\label{prop-Mer4}
We have 
\[
|\phi_{T-t}(v-i\alpha)|\leq C_1\exp\l\{-\frac{\sigma^2v^2(T-t)}{2}\r\}
\]
for any $v\in\bbR$, where
\begin{eqnarray*}
C_1
&=&\exp\l\{(T-t)\l[\alpha\mu^*+\frac{\sigma^2\alpha^2}{2}
   +\int_{\bbR_0}(e^{\alpha x}-1-\alpha x)\nu^{\bbP^*}(dx)\r]\r\} \\
&=&\exp\Bigg\{(T-t)\Bigg[\alpha\mu^*+\frac{\sigma^2\alpha^2}{2}+(1+h)\gamma
   (e^{m\alpha+\frac{\alpha^2\delta^2}{2}}-1-\alpha m) \\
&& \qquad -h\gamma e^{\frac{2m+\delta^2}{2}}\l[e^{(m+\delta^2)\alpha
   +\frac{\alpha^2\delta^2}{2}}-1-\alpha(m+\delta^2)\r]\Bigg]\Bigg\}.
\end{eqnarray*}
\end{prop}

\begin{proof}
Proposition~\ref{prop-Mer2} implies that
{\footnotesize
\begin{eqnarray*}
\lefteqn{\phi_{T-t}(v-i\alpha)} \\
&=&\exp\l\{(T-t)\l[i(v-i\alpha)\mu^*-\frac{\sigma^2(v-i\alpha)^2}{2}
   +\int_{\bbR_0}(e^{i(v-i\alpha)x}-1-i(v-i\alpha)x)\nu^{\tP}(dx)\r]\r\} \\
&=&\exp\l\{(T-t)\l[(iv+\alpha)\mu^*-\frac{\sigma^2(v^2-2i\alpha v-\alpha^2)}{2}
   +\int_{\bbR_0}(e^{(iv+\alpha)x}-1-(iv+\alpha)x)\nu^{\tP}(dx)\r]\r\} \\
&=&\exp\l\{(T-t)iv\l[\mu^*+\sigma^2\alpha-\int_{\bbR_0}x\nu^{\tP}(dx)\r]\r\}
   \exp\l\{(T-t)\int_{\bbR_0}e^{(iv+\alpha)x}\nu^{\tP}(dx)\r\} \\
&& \times\exp\l\{(T-t)\l[\alpha\mu^*-\frac{\sigma^2(v^2-\alpha^2)}{2}
   +\int_{\bbR_0}(-1-\alpha x)\nu^{\tP}(dx)\r]\r\}.
\end{eqnarray*}}
Noting that 
$\big|\exp\big\{(T-t)\int_{\bbR_0}e^{(iv+\alpha)x}\nu^{\tP}(dx)\big\}\big|
\leq\exp\big\{(T-t)\int_{\bbR_0}e^{\alpha x}\nu^{\tP}(dx)\big\}$, we have
{\small
\[
|\phi_{T-t}(v-i\alpha)|
\leq\exp\l\{(T-t)\l[\alpha\mu^*-\frac{\sigma^2(v^2-\alpha^2)}{2}
    +\int_{\bbR_0}(e^{\alpha x}-1-\alpha x)\nu^{\tP}(dx)\r]\r\}.
\]}
 \ 
\end{proof}

Propositions~\ref{prop-Mer5} and \ref{prop-Mer6} below give
sufficient conditions for $N\eta$ under which $I_1$ and $I_2$ satisfy
(\ref{eq-CM3}) for a given allowable error $\ve>0$, respectively.

%%%%%%%%%%%%%%%%%%%%%%%%%%%%%%%%%%%%%%%%%%%%%%%%%%%%%%%%%%%%%%%%%%%%%%%%%%%%%%%
\begin{prop}
\label{prop-Mer5}
Let $\ve>0$ and $t\in[0,T)$.
When $a>0$ satisfies
\begin{equation}
\label{eq-prop-Mer5}
\l(\frac{K}{\pi}\l(\frac{K}{S_{t-}}\r)^{-\alpha}C_1\r)^{1/4}
\frac{1}{\sigma\sqrt{T-t}\ve^{1/4}}
\leq a,
\end{equation}
we have
\[
\l|\frac{1}{\pi}\int_a^\infty\frac{K^{-iv-\alpha+1}}{\alpha-1+iv}
\phi_{T-t}(v-i\alpha)S_{t-}^{\alpha+iv}dv\r|
\leq \ve.
\]
\end{prop}

\begin{proof}
Noting that $e^{-x}\leq x^{-2}$ for any $x>0$,
we have, by Proposition~\ref{prop-Mer4},
\begin{eqnarray*}
\lefteqn{\l|\frac{1}{\pi}\int_a^\infty\frac{K^{-iv-\alpha+1}}{\alpha-1+iv}
         \phi_{T-t}(v-i\alpha)S_{t-}^{\alpha+iv}dv\r|} \\
&\qquad\leq& \frac{1}{\pi}\int_a^\infty\frac{K^{-\alpha+1}}{|\alpha-1+iv|}
       |\phi_{T-t}(v-i\alpha)|S_{t-}^\alpha dv  \\
&\qquad\leq& \frac{K}{\pi}\l(\frac{K}{S_{t-}}\r)^{-\alpha}\int_a^\infty\frac{1}
       {|\alpha-1+iv|}C_1e^{-\frac{\sigma^2v^2}{2}(T-t)}dv \\
&\qquad\leq& \frac{K}{\pi}\l(\frac{K}{S_{t-}}\r)^{-\alpha}C_1\int_a^\infty
       \frac{1}{v}\l\{\frac{\sigma^2v^2}{2}(T-t)\r\}^{-2}dv \\
&\qquad=&    \frac{K}{\pi}\l(\frac{K}{S_{t-}}\r)^{-\alpha}C_1\int_a^\infty
       \frac{4v^{-5}}{\sigma^4(T-t)^2}dv
       =\frac{K}{\pi}\l(\frac{K}{S_{t-}}\r)^{-\alpha}\frac{C_1}
       {\sigma^4(T-t)^2a^4} \\
&\qquad\leq& \ve\,.
\end{eqnarray*}
\end{proof}

%%%%%%%%%%%%%%%%%%%%%%%%%%%%%%%%%%%%%%%%%%%%%%%%%%%%%%%%%%%%%%%%%%%%%%%%%%%%%%%
\begin{prop}
\label{prop-Mer6}
Let $\ve>0$ and $t\in[0,T)$.
If $a>0$ satisfies
\begin{equation}
\label{eq-prop-Mer6-1}
\frac{4C_1\gamma K}{5\pi\sigma^4(T-t)^2\ve}\l(\frac{K}{S_{t-}}\r)^{-\alpha}
\l\{e^{(\alpha+1)m+(\frac{\alpha^2}{2}+\alpha+\frac{1}{2})\delta^2}
+e^{m\alpha+\frac{\delta^2\alpha^2}{2}}+|1-e^{m+\frac{\delta^2}{2}}|\r\}
\leq a^5,
\end{equation}
then
\begin{equation}
\label{eq-prop-Mer6-2}
\l|\frac{1}{\pi}\int_a^\infty K^{-i\zeta+1}
\int_{\bbR_0}(e^{i\zeta x}-1)(e^x-1)\nu(dx)\psi_2(\zeta)dv\r|<\ve\,.
\end{equation}
\end{prop}

\begin{proof}
First, we estimate $\int_a^\infty|\psi_2(\zeta)|dv$.
Noting that 
\[
\l|\frac{1}{(i\zeta-1)i\zeta}\r|=\l|\frac{1}{(iv+\alpha-1)(iv+\alpha)}\r|
\leq\frac{1}{v^2}\,,
\]
Proposition~\ref{prop-Mer4} implies
\begin{align*}
\int_a^\infty|\psi_2(\zeta)|dv
&=    \int_a^\infty\l|\frac{\phi_{T-t}(v-i\alpha)S_{t-}^{i(v-i\alpha)}}
       {(i\zeta-1)i\zeta}\r|dv
       \leq C_1S_{t-}^\alpha\int_a^\infty\frac{e^{-\frac{\sigma^2v^2}{2}(T-t)}}
       {v^2}\,dv \\
&\leq \frac{4C_1S_{t-}^\alpha}{\sigma^4(T-t)^2}\int_a^\infty v^{-6}\,dv
       =\frac{4C_1S_{t-}^\alpha}{5\sigma^4(T-t)^2a^5}\,.
\end{align*}
Hence, Proposition~\ref{prop-Mer3} implies that
{\small
\begin{eqnarray*}
\lefteqn{\mbox{L.H.S.\ of (\ref{eq-prop-Mer6-2})}} \\
&=&   \l|\frac{\gamma e^{2m+\frac{3}{2}\delta^2}}{\pi}
      \int_a^\infty(Ke^{-m-\delta^2})^{-i\zeta+1}\tilde{\psi}(\zeta)dv
      -\frac{\gamma e^m}{\pi}\int_a^\infty(Ke^{-m})^{-i\zeta+1}\tilde{\psi}
      (\zeta)dv \r. \\
&&    \l.+\,\frac{\gamma(1-e^{m+\frac{\delta^2}{2}})}{\pi}
      \int_a^\infty K^{-i\zeta+1}\psi_2(\zeta)dv\r| \\
&\leq&\frac{\gamma}{\pi}\int_a^\infty\l|e^{2m+\frac{3}{2}\delta^2}
      (Ke^{-m-\delta^2})^{-iv-\alpha+1}-e^m(Ke^{-m})^{-iv-\alpha+1}\r|
      \l|\psi_2(\zeta)\r|\l|e^{-\frac{\delta^2\zeta^2}{2}}\r|dv  \\
&&    +\,\frac{\gamma|1-e^{m+\frac{\delta^2}{2}}|}{\pi}
      \int_a^\infty|K^{-iv-\alpha+1}|\,|\psi_2(\zeta)|dv \\
&\leq&\frac{\gamma}{\pi}\int_a^\infty\l\{e^{2m+\frac{3}{2}\delta^2}
      (Ke^{-m-\delta^2})^{-\alpha+1}+e^m(Ke^{-m})^{-\alpha+1}\r\}
      |\psi_2(\zeta)|e^{-\frac{\delta^2(v^2-\alpha^2)}{2}}dv \\
&&    +\,\frac{\gamma|1-e^{m+\frac{\delta^2}{2}}|}{\pi}
      \int_a^\infty K^{-\alpha+1}|\psi_2(\zeta)|dv \\
&\leq&\frac{\gamma K^{-\alpha+1}}{\pi}\l\{e^{(\alpha+1)m+(\frac{\alpha^2}{2}
      +\alpha+\frac{1}{2})\delta^2}+e^{m\alpha+\frac{\delta^2\alpha^2}{2}}
      +|1-e^{m+\frac{\delta^2}{2}}|\r\}\int_a^\infty|\psi_2(\zeta)|dv \\
&\leq&\frac{4C_1\gamma K}{5\pi\sigma^4(T-t)^2a^5}
      \l(\frac{K}{S_{t-}}\r)^{-\alpha}\l\{e^{(\alpha+1)m+(\frac{\alpha^2}{2}
      +\alpha+\frac{1}{2})\delta^2}+e^{m\alpha+\frac{\delta^2\alpha^2}{2}}
      +|1-e^{m+\frac{\delta^2}{2}}|\r\} \\
&\leq&\ve\,.
\end{eqnarray*}}
\end{proof}

%%%%%%%%%%%%%%%%%%%%%%%%%%%%%%%%%%%%%%%%%%%%%%%%%%%%%%%%%%%%%%%%%%%%%%%%%%%%%%%
\subsection{Numerical results}
%%%%%%%%%%%%%%%%%%%%%%%%%%%%%%%%%%%%%%%%%%%%%%%%%%%%%%%%%%%%%%%%%%%%%%%%%%%%%%%
As seen in the previous subsection,
substituting (\ref{eq-I1}) and (\ref{eq-prop-Mer3}) for $I_1$ and $I_2$
respectively,
we can compute $LRM_t$ given in (\ref{eq-prop-AS}) with the FFT.
Note that we need Proposition~\ref{prop-Mer2}
in order to calculate $\psi_1$, $\psi_2$, and $\tilde{\psi}$.
In this subsection, we provide numerical results
for a Merton jump-diffusion model with parameters
$T=1$, $\mu=-0.7$, $\sigma=0.2$, $\gamma=1$, $m=0$, and $\delta=1$.
Note that $\mu^S$ is given by $-0.03$, which satisfies the second condition of
Assumption~\ref{ass-1}.
In particular, we consider the following two cases:
First, fixing the strike price $K$ to 1, we compute $LRM_t$
for times $t=0, 0.05,\ldots,0.95$.
Second, $t$ is fixed to $0.5$ and we instead vary $K$ from 1 to 8
at steps of 0.25 and compute $LRM_{0.5}$.
Note that we take $L_{t-}=1$ whatever the value of $t$ is taken.
Moreover, we choose $N=2^{14}$, $\eta=0.025$, nad $\alpha=1.75$
as parameters related to the FFT.
We have then $N\eta=409.6$.
For any parameter set mentioned above,
both (\ref{eq-prop-Mer5}) and (\ref{eq-prop-Mer6-1}) are satisfied
for $\epsilon=10^{-2}$.
Figure \ref{Fig-Mer} shows the results for these two cases.
The computation time to obtain Fig.~\ref{Fig-Mer}(b) was 0.59~s.
Note that all numerical experiments in this paper were carried out using
MATLAB (8.1.0.604 R2013a) on an Intel Core i7 3.4~GHz CPU with
16 GB 1333 MHz DDR3 memory.

%%%%%%%%%%%%%%%%%%%%%%%%%%%%%%%%%%%%%%%%%%%%%%%%%%%%%%%%%%%%%%%%%%%%%%%%%%%%%%%
\clearpage
\begin{figure}
\vspace{-5mm}
\centering
\subfigure[$LRM_t$ of a call option
           with strike price $K=1$ and maturity $T=1$ vs.\ time
           for a Merton jump-diffusion model with parameters $\mu=-0.7$,
           $\sigma=0.2$, $\gamma=1$, $m=0$, and $\delta=1$.
           ]{
           \includegraphics[width=0.9\textwidth]{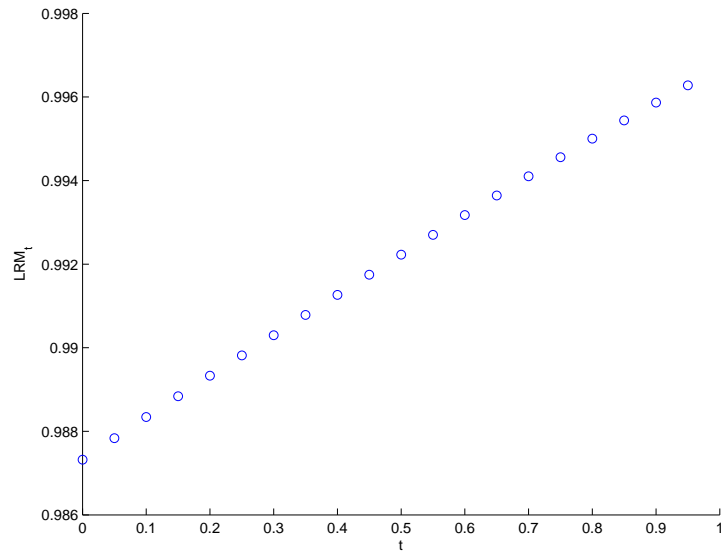}}
\subfigure[$LRM_{0.5}$ for the same Merton jump-diffusion model as (a)
           vs.\ strike price $K$.
           The vertical axis represents the value of $LRM_{0.5}$]{
           \includegraphics[width=0.9\textwidth]{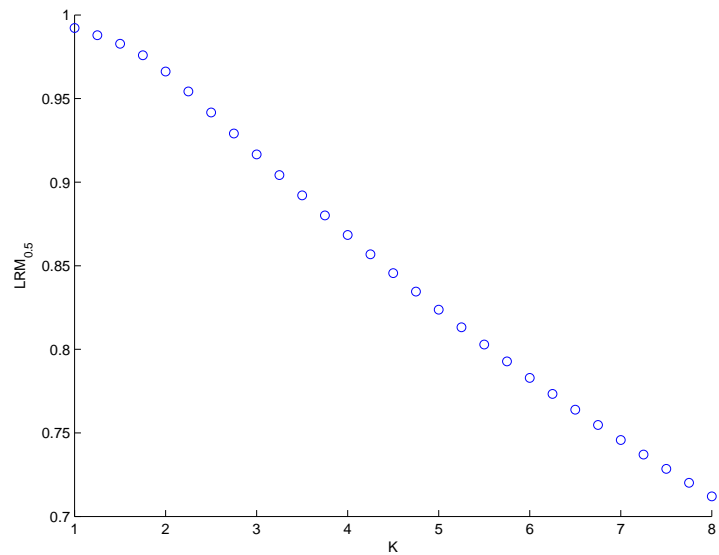}}
\caption{Merton jump-diffusion model}
\label{Fig-Mer}
\end{figure}
\clearpage

%%%%%%%%%%%%%%%%%%%%%%%%%%%%%%%%%%%%%%%%%%%%%%%%%%%%%%%%%%%%%%%%%%%%%%%%%%%%%%%
%                                                                             %
% section 4                                                                   %
%                                                                             %
%%%%%%%%%%%%%%%%%%%%%%%%%%%%%%%%%%%%%%%%%%%%%%%%%%%%%%%%%%%%%%%%%%%%%%%%%%%%%%%
\setcounter{equation}{0}
\section{Variance Gamma Models}
%%%%%%%%%%%%%%%%%%%%%%%%%%%%%%%%%%%%%%%%%%%%%%%%%%%%%%%%%%%%%%%%%%%%%%%%%%%%%%%
We now consider the case in which
$L$ is given as a variance gamma process.
Note that $L$ does not have a diffusion component.
This means that $\sigma=0$, that is, $I_1$ vanishes.
A variance gamma process,
which has three parameters $\kappa>0$, $m\in\bbR$, and $\delta>0$,
is defined as a time-changed Brownian motion with volatility $\delta$,
drift $m$, and subordinator $G_t$, where $G_t$ is a gamma process
with parameters $(1/\kappa, 1/\kappa)$.
In summary, $L$ is represented as
\[
L_t=mG_t+\delta B_{G_t}\ \ \mbox{ for }t\in[0,T]\,,
\]
where $B$ is a one-dimensional standard Brownian motion.
Moreover, the L\'evy measure of $L$ is given by
\[
\nu(dx)
= C(\mathbf{1}_{\{x<0\}}e^{-G|x|}+\mathbf{1}_{\{x>0\}}e^{-M|x|})\frac{dx}{|x|}
= C(\mathbf{1}_{\{x<0\}}e^{Gx}+\mathbf{1}_{\{x>0\}}e^{-Mx})\frac{dx}{|x|},
\]
where
\[
C:=\frac{1}{\kappa}, \qquad
G:=\frac{1}{\delta^2}\sqrt{m^2+\frac{2\delta^2}{\kappa}}+\frac{m}{\delta^2},
\qquad
M:=\frac{1}{\delta^2}\sqrt{m^2+\frac{2\delta^2}{\kappa}}-\frac{m}{\delta^2}.
\]
Note that $C$, $G$, and $M$ are positive.
To emphasize the parameters, we write $\nu$ with parameters $\kappa$, $m$, and
$\delta$ as $\nu(dx)=\nu[\kappa,m,\delta](dx)$.
Moreover, by regarding $C$, $G$, and $M$ as parameters, we may express $\nu$ as
$\nu(dx)=\nu_{C,G,M}(dx)$.
In addition, we assume $M>4$ in this section, which ensures that
the first condition of Assumption~\ref{ass-1} holds, by the following lemma:

%%%%%%%%%%%%%%%%%%%%%%%%%%%%%%%%%%%%%%%%%%%%%%%%%%%%%%%%%%%%%%%%%%%%%%%%%%%%%%%
\begin{lem}
When $M>4$,
$\int_{\bbR_0}(e^x-1)^n\nu(dx)<\infty$ for $n=2,4$.
\end{lem}

\begin{proof}
For $n=2,4$, we have
\begin{align*}
\int_1^\infty(e^x-1)^n\nu(dx)
&\leq C\int_1^\infty e^{(n-M)x}dx<\infty\,, \\
\int_0^1(e^x-1)^n\nu(dx)
&\leq \int_0^1x^n(e-1)^n\nu(dx)\leq C(e-1)^n<\infty\,, \\
\int_{-1}^0(e^x-1)^n\nu(dx)
&\leq \int_{-1}^0(-x)^n\nu(dx)\leq C\int_{-1}^0(-x)^{n-1}dx<\infty\,, \\
\int_{-\infty}^{-1}(e^x-1)^n\nu(dx)
&\leq \int_{-\infty}^{-1}\nu(dx)\leq C\int_1^\infty e^{-Gx}dx<\infty\,,
\end{align*}
because $n-M<0$, $0\leq e^x-1\leq x(e-1)$ whenever $x\in[0,1]$,
$1+x\leq e^x$ for any $x\in\bbR$, and $e^x\leq1$ if $x\leq0$.
\end{proof}

%%%%%%%%%%%%%%%%%%%%%%%%%%%%%%%%%%%%%%%%%%%%%%%%%%%%%%%%%%%%%%%%%%%%%%%%%%%%%%%
\begin{rem}
We can generalize this lemma to
$\int_{\bbR_0}|e^x-1|^a\nu(dx)<\infty$ for any $a\in[1,M)$.
\end{rem}

\noindent
Because $\mu=\int_{\bbR_0}x\nu(dx)$,
(\ref{eq-lem-VG4}) below implies that
the second condition of Assumption~\ref{ass-1} can rewritten as
\[
\log\l(\frac{(M-1)(G+1)}{(M-2)(G+2)}\r)>0\geq \log\l(\frac{MG}{(M-1)(G+1)}\r),
\]
which is equivalent to $-3<G-M\leq-1$.

%%%%%%%%%%%%%%%%%%%%%%%%%%%%%%%%%%%%%%%%%%%%%%%%%%%%%%%%%%%%%%%%%%%%%%%%%%%%%%%
\subsection{Mathematical preliminaries}
%%%%%%%%%%%%%%%%%%%%%%%%%%%%%%%%%%%%%%%%%%%%%%%%%%%%%%%%%%%%%%%%%%%%%%%%%%%%%%%
The approach to variance gamma models is similar to that in Subsection~3.1.
We begin by calculating of $\nu^{\tP}$.

%%%%%%%%%%%%%%%%%%%%%%%%%%%%%%%%%%%%%%%%%%%%%%%%%%%%%%%%%%%%%%%%%%%%%%%%%%%%%%%
\begin{prop}
\label{prop-VG0}
$\nu^{\tP}(dx)=\nu_{(1+h)C,G,M}(dx)+\nu_{-hC,G+1,M-1}(dx)$,
where $h=\frac{\mu^S}{\int_{\bbR_0}(e^x-1)^2\nu(dx)}$.
\end{prop}

\begin{proof}
By the same argument as Proposition~\ref{prop-Mer1}, 
$\nu^{\bbP^*}(dx)=(1+h)\nu(dx)-he^x\nu(dx)$.
We have $\lambda\nu_{C,G,M}(dx)=\nu_{\lambda C,G,M}(dx)$ for any $\lambda>0$,
and
\begin{align*}
e^x\nu_{C,G,M}(dx)
&= e^xC(\mathbf{1}_{\{x<0\}}e^{Gx}+\mathbf{1}_{\{x>0\}}e^{-Mx})\frac{dx}{|x|}
   \\
&= C(\mathbf{1}_{\{x<0\}}e^{(G+1)x}+\mathbf{1}_{\{x>0\}}e^{-(M-1)x})
   \frac{dx}{|x|}
&= \nu_{C,G+1,M-1}(dx)
\end{align*}
because $M-1>0$.
\end{proof}

%%%%%%%%%%%%%%%%%%%%%%%%%%%%%%%%%%%%%%%%%%%%%%%%%%%%%%%%%%%%%%%%%%%%%%%%%%%%%%%
\begin{rem}
For any $\lambda>0$, $\lambda\nu[\kappa,m,\delta](dx)$ is
a L\'evy measure corresponding to the variance gamma process
with parameters $\kappa/\lambda$, $\lambda m$, and $\delta\sqrt{\lambda}$.
However, $\nu_{C,G+1,M-1}(dx)$ is \textit{not} necessarily a L\'evy measure
corresponding to a variance gamma process.
\end{rem}

Next we calculate the characteristic function $\phi_{T-t}$ of $L$ under $\tP$:

%%%%%%%%%%%%%%%%%%%%%%%%%%%%%%%%%%%%%%%%%%%%%%%%%%%%%%%%%%%%%%%%%%%%%%%%%%%%%%%
\begin{prop}
\label{prop-VG1}
For any $t\in[0,T]$ and $v\in\bbR$, with $\zeta := v-i\alpha$, we have
{\small
\begin{eqnarray*}
\phi_{T-t}(\zeta)
&=& \l[\l(1+\frac{i\zeta}{G}\r)\l(1-\frac{i\zeta}{M}\r)\r]^{-(1 + h)(T-t)C}
    \l[\l(1+\frac{i\zeta}{G+1}\r)\l(1-\frac{i\zeta}{M-1}\r)\r]^{h(T-t)C} \\
&&  \times\exp\l\{(T-t)i\zeta\l[\mu^\ast+(1+h)C\frac{M-G}{GM}
    -hC\frac{M-G-2}{(G+1)(M-1)}\r]\r\},
\end{eqnarray*}}
where $\mu^\ast=\int_{\bbR_0}(x-e^x+1)\nu^{\tP}(dx)$.
\end{prop}

\begin{proof}
First of all, we have
\begin{align}
\int_0^\infty(e^{i\zeta x}-1)\frac{e^{-Mx}}{x}\,dx
&= \int_0^\infty\frac{e^{-(M-\alpha-iv)x}-e^{-Mx}}{x}\,dx \nonumber \\
&= \int_0^\infty\frac{e^{-(M-\alpha-iv)x}-e^{-(M-\alpha)x}
   +e^{-(M-\alpha)x}-e^{-Mx}}{x}\,dx \nonumber \\
&= i\int_0^\infty e^{-(M-\alpha)x}\int_0^ve^{itx}\,dt\,dx
   +\int_0^\infty\int_{M-\alpha}^Me^{-tx}\,dt\,dx \nonumber \\
&= i\int_0^v\int_0^\infty e^{-(M-\alpha-it)x}\,dx\,dt
   +\int_{M-\alpha}^M\int_0^\infty e^{-tx}\,dx\,dt \nonumber \\
\label{Frullani}
&= \log\l(\frac{M-\alpha}{M-\alpha-iv}\r)+\log\l(\frac{M}{M-\alpha}\r)
= -\log\l(1-\frac{i\zeta}{M}\r),
\end{align}
which provides
\begin{align*}
\int_{\bbR_0}(e^{i\zeta x}-1)\nu_{C,G,M}(dx)
&= C\int_{-\infty}^0(e^{i\zeta x}-1)\frac{e^{Gx}}{-x}\,dx
   +C\int_0^\infty(e^{i\zeta x}-1)\frac{e^{-Mx}}{x}\,dx \\
&= -C\l(\log\l(1+\frac{i\zeta}{G}\r)+\log\l(1-\frac{i\zeta}{M}\r)\r).
\end{align*}
In addition, we have
\[
\int_{\bbR_0}x\nu_{C,G,M}(dx)
= -C\int_{-\infty}^0e^{Gx}\,dx+C\int_0^\infty e^{-Mx}\,dx
= -C\frac{M-G}{GM}.
\]
Together with Proposition~\ref{prop-VG0}, we obtain
\begin{eqnarray*}
\int_{\bbR_0}(e^{i\zeta x}-1-i\zeta x)\nu^{\tP}(dx)
&=& \log\l(1+\frac{i\zeta}{G}\r)^{-(1+h)C}
    +\log\l(1-\frac{i\zeta}{M}\r)^{-(1+h)C} \\
&&  +\log\l(1+\frac{i\zeta}{G+1}\r)^{hC}+\log\l(1-\frac{i\zeta}{M-1}\r)^{hC} \\
&&  +\,i(1+h)C\zeta\frac{M-G}{GM}-ihC\zeta\frac{M-G-2}{(G+1)(M-1)},
\end{eqnarray*}
from which Proposition~\ref{prop-VG1} follows.
\end{proof}

Now, we reformulate (\ref{eq-I2-2}) into a linear combination of
two Fourier transforms in order to allow use of the FFT.
As preparation, we show the following:

%%%%%%%%%%%%%%%%%%%%%%%%%%%%%%%%%%%%%%%%%%%%%%%%%%%%%%%%%%%%%%%%%%%%%%%%%%%%%%%
\begin{lem}
\label{lem-VG4}
\begin{equation}
\label{eq-lem-VG4}
\int_{\bbR_0}e^{i\zeta x}(e^x-1)\nu(dx)
=C\log\l(\frac{M-i\zeta}{M-1-i\zeta}\frac{G+i\zeta}{G+1+i\zeta}\r).
\end{equation}
\end{lem}

\begin{proof}
First of all, we have
\begin{align}
\lefteqn{\int_{\bbR_0}e^{i\zeta x}(e^x-1)\nu(dx)} \nonumber \\
&= \int_{\bbR_0}e^{(iv+\alpha)x}(e^x-1)\nu(dx) \nonumber \\
&= C\l\{\int_0^\infty\frac{1-e^x}{x}e^{-(G+\alpha+1+iv)x}\,dx
    +\int_0^\infty\frac{e^x-1}{x}e^{-(M-\alpha-iv)x}\,dx\r\}.
\label{eq-lem-VG4-1}
\end{align}
To calculate (\ref{eq-lem-VG4-1}), we compute
\[
\int_0^\infty\frac{e^x-1}{x}e^{-ax}\cos{bx}\,dx
\qquad\mbox{ and }\qquad
\int_0^\infty\frac{e^x-1}{x}e^{-ax}\sin{bx}\,dx
\]
for $a>1$ and $b\in\bbR$.
First, we have
\begin{eqnarray}
\lefteqn{\int_0^\infty\frac{e^x-1}{x}e^{-ax}\cos{bx}\,dx} \nonumber \\
&=& \int_0^\infty\frac{\cos{bx}}{x}\int_{a-1}^axe^{-tx}\,dt\,dx
    =\int_{a-1}^a\int_0^\infty\cos{bx}\cdot e^{-tx}\,dx\,dt \nonumber \\
&=& \int_{a-1}^a\frac{t}{t^2+b^2}dt
    =\frac{1}{2}\log\l(\frac{a^2+b^2}{(a-1)^2+b^2}\r).
\label{eq-lem-VG4-2}
\end{eqnarray}
A similar calculation implies that
\begin{equation}
\label{eq-lem-VG4-3}
\int_0^\infty\frac{e^x-1}{x}e^{-ax}\sin{bx}\,dx
=\int_{a-1}^a\frac{b}{t^2+b^2}\,dt=\tan^{-1}\frac{a}{b}-\tan^{-1}\frac{a-1}{b}.
\end{equation}
Noting that $M-\alpha>2$ and $\tan^{-1}x=\frac{i}{2}\log\frac{i+x}{i-x}$
for $x\in\bbR$, we have, by (\ref{eq-lem-VG4-2}) and (\ref{eq-lem-VG4-3}),
\begin{align}
\lefteqn{\int_0^\infty\frac{e^x-1}{x}e^{-(M-\alpha-iv)x}\,dx} \nonumber \\
&\qquad= \int_0^\infty\frac{e^x-1}{x}e^{-(M-\alpha)x}\cos{vx}\,dx
    +i\int_0^\infty\frac{e^x-1}{x}e^{-(M-\alpha)x}\sin{vx}\,dx \nonumber \\
&\qquad= \frac{1}{2}\log\l(\frac{(M-\alpha)^2+v^2}{(M-\alpha-1)^2+v^2}\r)
    +i\l(\tan^{-1}\frac{M-\alpha}{v}-\tan^{-1}\frac{M-\alpha-1}{v}\r)
    \nonumber \\
&\qquad= \log\l(\frac{M-\alpha-iv}{M-\alpha-1-iv}\r).
\label{eq-lem-VG4-4}
\end{align}
Calculating the first term of the right-hand side of (\ref{eq-lem-VG4-1})
in the same way as the above, we obtain
\begin{equation}
\label{eq-lem-VG4-5}
\int_0^\infty\frac{1-e^x}{x}e^{-(G+\alpha+1+iv)x}\,dx
=\log\l(\frac{G+\alpha+iv}{G+\alpha+1+iv}\r).
\end{equation}
Substituting (\ref{eq-lem-VG4-4}) and (\ref{eq-lem-VG4-5}) for
(\ref{eq-lem-VG4-1}), we arrive at (\ref{eq-lem-VG4}).
\end{proof}

\noindent
From the above lemma, $I_2$ is given as follows:
{\small 
\begin{eqnarray}
\lefteqn{\int_{\bbR_0}\bbE_{\tP}[(S_Te^x-K)^+-(S_T-K)^+\mid\calF_{t-}]
(e^x-1)\nu(dx)} \nonumber \\
&\qquad=& \frac{1}{\pi}\int^\infty_0K^{-i\zeta+1}
    \int_{\bbR_0}(e^{i\zeta x}-1)(e^x-1)\nu(dx)\psi_2(\zeta)dv \nonumber \\
&\qquad=& \frac{1}{\pi}\int^\infty_0K^{-i\zeta+1}\widetilde\psi_{VG}(\zeta)dv
    -\frac{1}{\pi}\int^\infty_0 C\log\l(\frac{MG}{(M-1)(G+1)}\r)K^{-i\zeta+1}
    \psi_2(\zeta)dv. \nonumber \\
\label{eq-VG-I2}
\end{eqnarray}}
where
\[
\widetilde\psi_{VG}(\zeta):=
C\log\l(\frac{M-i\zeta}{M-1-i\zeta}\frac{G+i\zeta}{G+1+i\zeta}\r)\psi_2(\zeta).
\]
Recall that
$\psi_2(\zeta)=\frac{\phi_{T-t}(\zeta)S_{t-}^{i\zeta}}{(i\zeta-1)i\zeta}$.
As a result, we need only use the FFT twice for computing $I_2$.

As the final item of this subsection,
we estimate a sufficient length for the integration interval of
(\ref{eq-VG-I2}) for a given allowable error $\ve>0$
in the sense of (\ref{eq-CM3}).
We first provide an upper estimate of $\phi_{T-t}$ as follows:

%%%%%%%%%%%%%%%%%%%%%%%%%%%%%%%%%%%%%%%%%%%%%%%%%%%%%%%%%%%%%%%%%%%%%%%%%%%%%%%
\begin{prop}
\label{prop-VG2}
For any $v\in\bbR$,
\[
|\phi_{T-t}(v-i\alpha)|\leq C_2|v|^{-2C(T-t)},
\]
where
\begin{eqnarray}
\label{eq-prop-VG2}
C_2 &=& (GM)^{(1+h)(T-t)C}[(G+1)(M-1)]^{-h(T-t)C} \nonumber \\
    &&  \times\exp\l\{(T-t)\alpha\l[\mu^\ast+(1+h)C\frac{M-G}{GM}
        -hC\frac{M-G-2}{(G+1)(M-1)}\r]\r\}.
\end{eqnarray}
\end{prop}

\begin{proof}
This can be seen because
\[
\l|\l(1+\frac{iv+\alpha}{G}\r)^{-a}\r|\leq\frac{G^a}{|v|^a}
\]
for any $a>0$.
\end{proof}

\noindent
We need to prepare one more lemma:

%%%%%%%%%%%%%%%%%%%%%%%%%%%%%%%%%%%%%%%%%%%%%%%%%%%%%%%%%%%%%%%%%%%%%%%%%%%%%%%
\begin{lem}
\label{lem-VG5}
\begin{equation}
\label{eq-lem-VG5}
\l|\int_{\bbR_0}e^{i\zeta x}(e^x-1)\nu(dx)\r|
\leq C\l\{\frac{1}{G+\alpha}+\frac{1}{M-\alpha-1}\r\}.
\end{equation}
\end{lem}

\begin{proof}
The same sort of calculations as in (\ref{Frullani}) imply
\begin{eqnarray*}
\lefteqn{\l|\int_{\bbR_0}e^{i\zeta x}(e^x-1)\nu(dx)\r|} \\
&\qquad\leq& C\l\{\l|\int_0^\infty\frac{1-e^x}{x}e^{-(G+\alpha+1+iv)x}\,dx\r|
       +\l|\int_0^\infty\frac{e^x-1}{x}e^{-(M-\alpha-iv)x}\,dx\r|\r\} \\
&\qquad\leq& C\l\{\int_0^\infty\frac{e^x-1}{x}e^{-(G+\alpha+1)x}\,dx
       +\int_0^\infty\frac{e^x-1}{x}e^{-(M-\alpha)x}\,dx\r\} \\
&\qquad=&    C\l\{\log\l(1+\frac{1}{G+\alpha}\r)
       +\log\l(1+\frac{1}{M-\alpha-1}\r)\r\} \\
&\qquad\leq& C\l\{\frac{1}{G+\alpha}+\frac{1}{M-\alpha-1}\r\}.
\end{eqnarray*}
\end{proof}

\noindent
When we calculate (\ref{eq-VG-I2}),
$N$ and $\eta$ should be taken so that $N\eta$ satisfies (\ref{eq-prop-VG3-1})
below for a given allowable error $\ve>0$.

%%%%%%%%%%%%%%%%%%%%%%%%%%%%%%%%%%%%%%%%%%%%%%%%%%%%%%%%%%%%%%%%%%%%%%%%%%%%%%%
\begin{prop}
\label{prop-VG3}
Let $\ve>0$.
When $a>0$ satisfies
{\small
\begin{align}
\label{eq-prop-VG3-1}
\frac{C C_2 K^{-\alpha+1}S_{t-}^\alpha}{\pi\ve(2C(T-t)+1)}
\l[\frac{1}{G+\alpha}+\frac{1}{M-\alpha-1}
+\l|\log\l(\frac{MG}{(M-1)(G+1)}\r)\r|\r]
\leq a^{2C(T-t)+1},
\end{align}}
we have
\begin{equation}
\label{eq-prop-VG3-2}
\l|\frac{1}{\pi}\int^\infty_aK^{-i\zeta+1}
\int_{\bbR_0}(e^{i\zeta x}-1)(e^x-1)\nu(dx)\psi_2(\zeta)dv\r|
<\ve\,,
\end{equation}
where $C_2$ is defined in (\ref{eq-prop-VG2}).
\end{prop}

\begin{proof}
By (\ref{eq-lem-VG5}), we have
{\small
\begin{eqnarray}
\lefteqn{\l|\frac{1}{\pi}\int^\infty_aK^{-i\zeta+1}
\int_{\bbR_0}(e^{i\zeta x}-1)(e^x-1)\nu(dx)\psi_2(\zeta)dv\r|} \nonumber \\
&\leq& \frac{1}{\pi}\Bigg\{\l|\int^\infty_aK^{-i\zeta+1}
       \int_{\bbR_0}e^{i\zeta x}(e^x-1)\nu(dx)\psi_2(\zeta)dv\r| \nonumber \\
&&     +\l|\int^\infty_aK^{-i\zeta+1}\int_{\bbR_0}(e^x-1)\nu(dx)
       \psi_2(\zeta)dv\r|\Bigg\} \nonumber \\
&\leq& \frac{1}{\pi}\l\{\int^\infty_a\l|K^{-i\zeta+1}\r|
       \l[\l|\int_{\bbR_0}e^{i\zeta x}(e^x-1)\nu(dx)\r|
       +\l|\int_{\bbR_0}(e^x-1)\nu(dx)\r|\r]\l|\psi_2(\zeta)\r|dv\r\}
       \nonumber \\
&\leq& \frac{1}{\pi}\l\{K^{-\alpha+1}C\l(\frac{1}{G+\alpha}
       +\frac{1}{M-\alpha-1}+\l|\log\l(\frac{MG}{(M-1)(G+1)}\r)\r|\r)
       \int_a^\infty\l|\psi_2(\zeta)\r|dv\r\}. \nonumber \\
\label{eq-prop-VG3-3}
\end{eqnarray}}
Because Proposition~\ref{prop-VG2} implies
\[
\l|\psi_2(\zeta)\r|
=    \l|\frac{\phi_{T-t}(\zeta)S_{t-}^{i\zeta}}{(i\zeta-1)i\zeta}\r|
\leq \frac{1}{v^2}C_2|v|^{-2C(T-t)}S_{t-}^\alpha
=    C_2S_{t-}^\alpha|v|^{-2C(T-t)-2},
\]
we have, together with (\ref{eq-prop-VG3-3}),
{\small
\begin{eqnarray*}
\mbox{R.H.S. of (\ref{eq-prop-VG3-2})}
&\leq& \frac{1}{\pi}C C_2K^{-\alpha+1}S_{t-}^\alpha
       \l[\frac{1}{G+\alpha}+\frac{1}{M-\alpha-1}
       +\l|\log\l(\frac{MG}{(M-1)(G+1)}\r)\r|\r] \\
&&     \times\int_a^\infty|v|^{-2C(T-t)-2}dv \\
&=&    \frac{1}{\pi}C C_2K^{-\alpha+1}S_{t-}^\alpha
       \l[\frac{1}{G+\alpha}+\frac{1}{M-\alpha-1}
       +\l|\log\l(\frac{MG}{(M-1)(G+1)}\r)\r|\r] \\
&&     \times\frac{a^{-2C(T-t)-1}}{2C(T-t)+1}.
\end{eqnarray*}}
\end{proof}

%%%%%%%%%%%%%%%%%%%%%%%%%%%%%%%%%%%%%%%%%%%%%%%%%%%%%%%%%%%%%%%%%%%%%%%%%%%%%%%
\subsection{Numerical results}
%%%%%%%%%%%%%%%%%%%%%%%%%%%%%%%%%%%%%%%%%%%%%%%%%%%%%%%%%%%%%%%%%%%%%%%%%%%%%%%
We illustrate our numerical results for a variance gamma model.
Choosing the model parameters as $\kappa=0.15$, $m=-0.2$, and $\delta=0.45$,
which meet the second condition of Assumption~\ref{ass-1},
we compute $LRM_t$ for the same numerical experiments as in Subsection~3.2.
Note that $M>4$ is satisfied.
Moreover, we also take the same parameters related to the FFT as
in Subsection~3.2.
$N\eta$ satisfies (\ref{eq-prop-VG3-1}) for any parameter set.
The results are shown in Fig.~\ref{Fig-VG}.
The computation time to obtain Fig.~\ref{Fig-VG}(b) was 0.19~s.

In addition, we implemented the same type of numerical experiments as the above
based on market data.
We used the Nikkei 225 index for March 2014.
We need to set the log price $L_t:=\log(S_t/S_0)$,
where $S_0$ is the price on 28 February 2014, which was 14841.07.
We estimate the parameters $C$, $G$, and $M$ in Table~\ref{tab1} from
the mean, variance, and skewness of the log price
by using the generalized method of moments and the Levenberg--Marquardt method.

%%%%%%%%%%%%%%%%%%%%%%%%%%%%%%%%%%%%%%%%%%%%%%%%%%%%%%%%%%%%%%%%%%%%%%%%%%%%%%%
\begin{table}[hbtp]
  \caption{Estimated parameters}
  \label{tab1}
  \begin{center}
    \begin{tabular}{cc}
      \hline 
      C & $2.469395026815120$  \\
      G & $23.743109051760964$ \\
      M & $24.903251787154687$ \\
      \hline
    \end{tabular}
  \end{center}
\end{table}

\noindent
Because $G-M\approx-1.16$, this parameter set satisfies Assumption~\ref{ass-1}.
We take $T=1$ and $S_{t-}=14841.07$, that is, $L_{t-}=0$.
First, fixing the strike price $K=14000$, we compute $LRM_t$
for $t=0, 0.05,\ldots,0.95$.
Next, fixing $t$ to $0.5$, the values of $LRM_{0.5}$ are computed
for $K=10000,11000,\dots,20000$.
Note that $N\eta$ satisfies (\ref{eq-prop-VG3-1}).
The results of the computation are illustrated in Fig.~\ref{Fig-VG2}.

%%%%%%%%%%%%%%%%%%%%%%%%%%%%%%%%%%%%%%%%%%%%%%%%%%%%%%%%%%%%%%%%%%%%%%%%%%%%%%%
\clearpage
\begin{figure}
\vspace{-2mm}
\centering
\subfigure[Results of the same numerical experiments
           as Fig.~\ref{Fig-Mer}(a) for a variance gamma model
           with parameters $\kappa=0.15$, $m=-0.2$, and $\delta=0.45$.]{
           \includegraphics[width=0.9\textwidth]{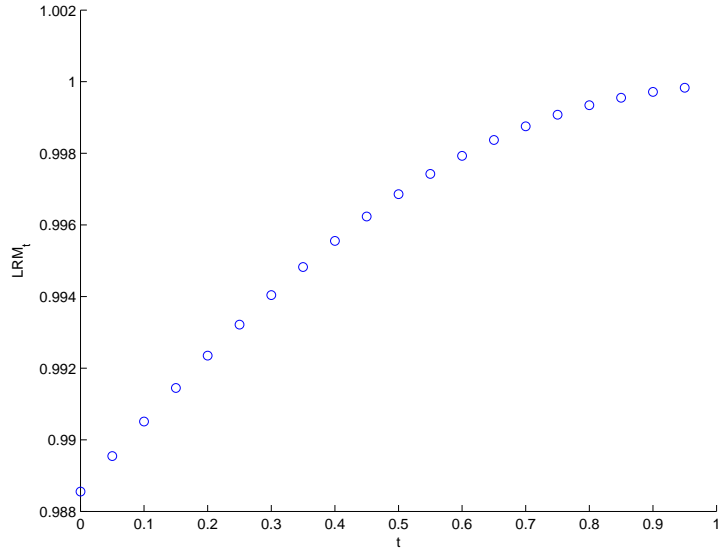}}
\subfigure[Results for the computation as
           Fig.~\ref{Fig-Mer}(b) for the same variance gamma model
           as in (a).]{
           \includegraphics[width=0.9\textwidth]{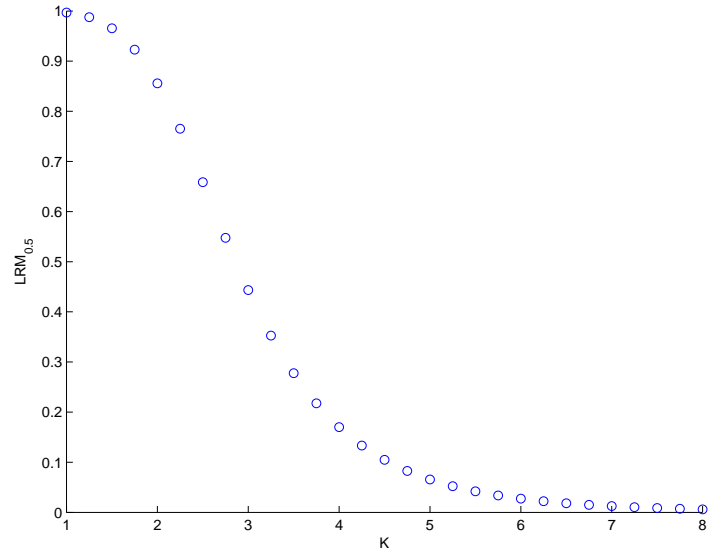}}
\caption{Variance gamma model with parameters $\kappa=0.15$, $m=-0.2$,
         $\delta=0.45$}
\label{Fig-VG}
\end{figure}
\clearpage
%%%%%%%%%%%%%%%%%%%%%%%%%%%%%%%%%%%%%%%%%%%%%%%%%%%%%%%%%%%%%%%%%%%%%%%%%%%%%%%
\clearpage
\begin{figure}
\vspace{-2mm}
\centering
\subfigure[Values of $LRM_t$ with strike price
           $K=14000$ and $S_{t-}=14841.07$ for $t=0, 0.05,\ldots,0.95$.
           The values of the three parameters $C$, $G$, and $M$ are given
           in Table~\ref{tab1}.]{
           \includegraphics[width=0.9\textwidth]{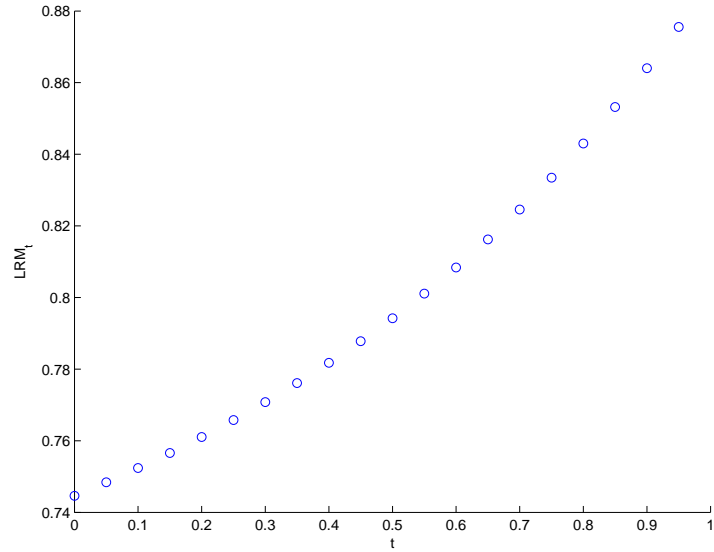}}
\subfigure[$LRM_{0.5}$ for $K=10000,11000,\dots,20000$
           with $S_{0.5}=14841.07$.]{
           \includegraphics[width=0.9\textwidth]{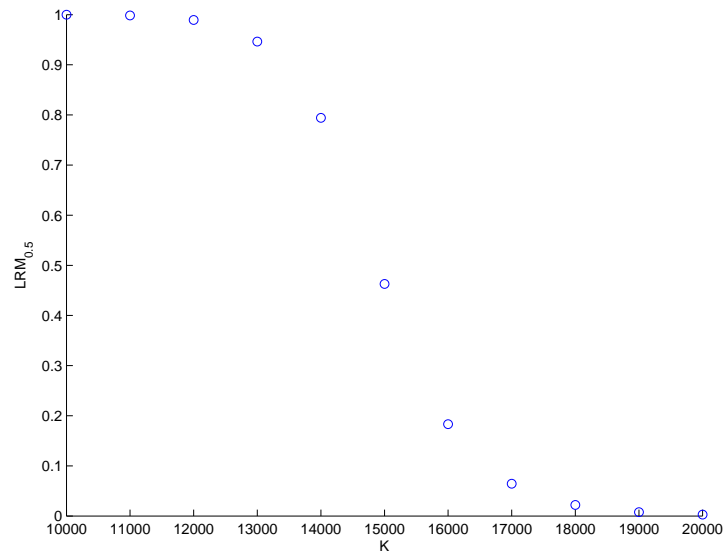}}
\caption{Variance gamma model based on the Nikkei 225 index for March 2014}
\label{Fig-VG2}
\end{figure}
\clearpage

%%%%%%%%%%%%%%%%%%%%%%%%%%%%%%%%%%%%%%%%%%%%%%%%%%%%%%%%%%%%%%%%%%%%%%%%%%%%%%%
\begin{center}
{\bf Acknowledgements}
\end{center}
Takuji Arai gratefully acknowledges financial support from the
Ishii Memorial Securities Research Promotion Foundation.

%%%%%%%%%%%%%%%%%%%%%%%%%%%%%%%%%%%%%%%%%%%%%%%%%%%%%%%%%%%%%%%%%%%%%%%%%%%%%%%

%%%%%%%%%%%%%%%%%%%%%%%%%%%%%%%%%%%%%%%%%%%%%%%%%%%%%%%%%%%%%%%%%%%%%%%%%%%%%%%
\end{document}